\documentclass[12pt]{iopart}

\usepackage[english]{babel}
\usepackage[utf8]{inputenc}
\usepackage[colorinlistoftodos, color=green!40, prependcaption]{todonotes}
\usepackage{color}
\usepackage{amsthm}
\usepackage{mathtools}
\usepackage{amsmath}
\usepackage{mathrsfs}
\usepackage{physics} 
\usepackage{xcolor}
\usepackage{graphicx}
\usepackage{soul}
\usepackage{bm}
\usepackage{comment}
\usepackage{MnSymbol}%
\usepackage{wasysym}
\usepackage{cite}
\usepackage{mathbbol}
\usepackage[pdftex, pdftitle={Article}, pdfauthor={Author}]{hyperref} 

\DeclareMathOperator*{\argmin}{arg\,min}
\newtheorem{theorem}{Theorem}
\theoremstyle{definition}
\newtheorem{remark}[theorem]{Remark}

\newtheorem{proposition}[theorem]{Proposition}

\begin{document}
\title{Progress towards analytically optimal angles in quantum approximate optimisation} 

\author{D.~Rabinovich$^*$, R.~Sengupta, E.~Campos, V.~Akshay, and J.~Biamonte}
\address{Skolkovo Institute of Science and Technology, 3 Nobel Street, Moscow, Russian Federation 121205}
\ead{$^*$daniil.rabinovich@skoltech.ru}
    
\begin{abstract}
The Quantum Approximate Optimisation Algorithm is a $p$ layer, time variable split operator method executed on a quantum processor and driven to convergence by classical outer loop optimisation.  The classical co-processor varies individual application times of a problem/driver propagator sequence to prepare a state which approximately minimizes the problem's generator.  Analytical solutions to choose optimal application times (called angles) have proven difficult to find, whereas outer loop optimisation is resource intensive.  Here we prove that optimal Quantum Approximate Optimisation Algorithm parameters for $p=1$ layer reduce to one free variable and in the thermodynamic limit, we recover optimal angles.  We moreover demonstrate that conditions for vanishing gradients of the overlap function share a similar form which leads to a linear relation between circuit parameters, independent on the number of qubits. Finally, we present a list of numerical effects, observed for particular system size and circuit depth, which are yet to be explained analytically. 
\end{abstract}

\section{Introduction}
The field of quantum algorithms has dramatically transformed in the last few years due to the advent of a quantum to classical feedback loop: a fixed depth quantum circuit is adjusted to minimize a cost function.  This approach partially circumvents certain limitations such as variability in pulse timing and requires shorter depth circuits at the cost of outer loop training \cite{harrigan2021quantum,pagano2019quantum, guerreschi2019qaoa,butko2020understanding, Biamonte_2021, campos2020abrupt}. The most studied algorithm in this setting is the Quantum Approximate Optimisation Algorithm (QAOA) \cite{Farhi2014} which was developed to approximate solutions to combinatiorial optimisation problem instances \cite{niu2019optimizing,Farhi2014,lloyd2018quantum,morales2020universality,Zhou2020,wang2020x,Brady2021,Farhi2016,Akshay2020,Farhi2019a,Wauters2020,Claes2021,Zhou}.

The setting of QAOA is that of $n$ qubits: states are represented as vectors in $V_n = [\mathbb{C}^2]^{\otimes n}$.  We are given a non-negative Hamiltonian $P \in \text{herm}_\mathbb{C}(V_n)$ and we seek the normalized ground vector $\ket{t}\in \argmin\limits_{\phi \in \{0,1\}^n} \bra{\phi}P\ket{\phi}$.

QAOA might be viewed as a (time variable fixed depth) quantum split operator method.  We let ${\mathcal V}(\gamma)$ be the propagator of $P$ applied for time $\gamma$.  We consider a second propagator ${\mathcal U}(\beta)$ generated by applying a yet to be defined Hamiltonian $H_x$ for time $\beta$.  We start off in the equal superposition state $\ket{+}^{\otimes{n}}=2^{-n/2}(\ket{0}+\ket{1})^{\otimes n}$ and form a $p$-depth ${\mathcal U}$, ${\mathcal V}$ sequence:
\begin{align}
    |g_p({\bf \gamma}, {\bf \beta})|^2 = |\bra{t}\Pi_{k=1}^p [{\mathcal U}(\beta_k){\mathcal V}(\gamma_k)]\ket{+}^{\otimes n}|^2. 
\end{align}
The time of application of each propagator is varied to maximize preparation of the state $\ket{t}$.  Finding ${\bf \gamma}, {\bf \beta}$ to maximize $|g_p({\bf \gamma}, {\bf \beta})|$ has shown to be cumbersome.  Even lacking such solutions, much progress has been made.  

Recent milestones include experimental demonstration of $p=3$ depth QAOA (corresponding to six tunable parameters) using twenty three qubits \cite{harrigan2021quantum}, universality results \cite{lloyd2018quantum,morales2020universality}, as well as several results that aid and improve on the original implementation of the algorithm \cite{Zhou2020,wang2020x,Brady2021}. Although QAOA exhibits provable advantages such as recovering a near optimal query complexity in Grover's search \cite{Jiang2017a} and offers a pathway towards quantum advantage \cite{Farhi2016}, several limitations have been discovered for low depth QAOA \cite{Akshay2020,hastings2019classical,Bravyi2019}. 

In the setting of Maximum-Constraint-Satifiability (e.g.~minimizing a Hamiltonian representing a function of type $f:\{0, 1\}^n\rightarrow \mathbb{R}_+$), it has been shown that underparameterisation of QAOA sequences can be induced by increasing a problem instances constraint to variable ratio \cite{Akshay2020}. This effect persists in graph minimisation problems \cite{Akshay_2021_google}. While this effect is perhaps an expected limitation of the quantum algorithm, parameter concentrations and noise assisted training add a degree of optimism.  QAOA exhibits parameter concentrations, in which training for some fraction of $\omega<n$ qubits provides a training sequence for $n$ qubits \cite{akshay2021parameter}.  Moreover, whereas layerwise training saturates for QAOA in which the algorithm plateaus and fails to reach the target, local coherent noise recovers layerwise trainings robustness \cite{campos2021training}.  Both concentrations and noise assisted training imply a reduction in computational resources required in outerloop optimisation. 

Exact solutions to find the optimal parameters for QAOA have only been possible in special cases including, e.g.~fully connected graphs \cite{Farhi2019a, Wauters2020,Claes2021} and projectors \cite{akshay2021parameter}. A general analytical approach that allows for (i) calculation of optimal parameters, (ii) estimation of the critical circuit depth and (iii) performance guarantees for fixed depth remains open.

Here we prove that optimal QAOA parameters for $p=1$ are related as $\gamma_1 = \pi - 2\beta_1$ and in the thermodynamic limit, we recover optimality as $\beta_1 n \rightarrow \pi$ and $\gamma_1 \rightarrow \pi$.  We moreover demonstrate that conditions for vanishing gradients of the overlap function share a similar form which leads to a linear relation between circuit parameters, independent of the number of qubits. We hence devise an additional means to recover parameter concentrations~\cite{akshay2021parameter} analytically. Finally, we present a list of numerical effects, observed for particular system size and circuit depth, which are yet to be explained analytically.


\section{Quantum Approximate Optimisation Algorithm}
\label{description}
We consider an $n$-qubit complex vector space $V_n = [\mathbb{C}^2]^{\otimes n}\cong \mathbb{C}^{2^n}$ with fixed standard computational basis $B_n=\{\ket{0},\ket{1}\}^{\otimes n}$. For an arbitrary target state $\ket{t}\in B_n$ (equivalently $\ket{t}, t\in \{0,1\}^{\times n}$) we define propagators
\begin{align}
    \mathcal U(\beta)\equiv e^{-i\beta H_x},~
  \mathcal V(\gamma)\equiv e^{-i\gamma P},
  \end{align}
where $P=\ketbra{t}{t}$ and ${H}_x = \sum_{j=1}^{n} X_{j}$ is the one-body mixer Hamiltonian with $X_j$ the Pauli matrix acting non-trivially on the $j$-th qubit.

A $p$-depth ($p$ layer) QAOA circuit prepares a quantum state $\ket{\psi}$ as:
\begin{align}
    \ket{\psi_p(\bm\gamma,\bm\beta)} =  \prod\limits_{k=1}^p [{\mathcal U}(\beta_k){\mathcal V}(\gamma_k)]\ket{+}^{\otimes{n}},
    \label{ansatz}
\end{align} 
where $\gamma_k\in[0,2\pi)$, $\beta_k\in[0,\pi)$.
The optimisation task is to determine QAOA optimal parameters for which the state prepared in \eqref{ansatz} achieves maximum absolute value of the overlap $g_p(\bm \gamma, \bm \beta)=\braket{t}{\psi_p(\bm \gamma, \bm \beta)}$  with the target $\ket{t}$. In other words, we search for 
\begin{align}
    (\bm \gamma_{opt},\bm \beta_{opt}) \in \arg\max_{\bm \gamma, \bm \beta} \abs{g_p(\bm \gamma,\bm \beta)}
    \label{argmax}
\end{align}
Note that the problem is equivalent to the minimization of the ground state energy of Hamiltonian $P^{\perp}  = \mathbb{1} - \ketbra{t}{t}$,
\begin{align}
    \min_{\bm\gamma,\bm\beta} \bra{\psi_p ( \bm\gamma, \bm\beta)} P^{\perp} \ket{\psi_p (\bm\gamma, \bm\beta)} = 1 - \max_{\bm\gamma,\bm\beta} \abs{g_p(\bm \gamma, \bm \beta)}^2.
    \label{equival}
\end{align}

\begin{remark}[Inversion symmetry]
Under the affine transformation 
\begin{align}
    (\gamma,\beta)\to(2\pi-\gamma,\pi-\beta)
    \label{affine}
\end{align} 
the absolute value of the overlap remains invariant as $g_p\to (-1)^n g_p^*$. Therefore, this narrows the search space to $\gamma_k\in[0,\pi)$, $\beta_k\in[0,\pi)$, whereas maximums inside the restricted region determine maximums in the composite space using Eq. \eqref{affine}.
\end{remark}

\begin{proposition}[Overlap invariance]
\label{invar}
The overlap function $g_p(\bm\gamma, \bm\beta)$ is invariant with respect to $\ket{t}\in B_n$. 
\end{proposition}
\begin{proof}
Each $\ket{t}=\ket{t_1t_2\dots t_n}\in B_n$ determines a unitary operator $U=U^\dagger=\bigotimes_{j=1}^{n}X_j^{t_j}.$ Hence, we have 
\begin{align}
    g_p(\bm\gamma,\bm\beta)=& \bra{\bm 0}U^\dagger \prod\limits_{k=1}^p e^{-i \beta_k  H_x} e^{-i \gamma_k U(\ketbra{\bm 0}{\bm 0})U^\dagger}\ket{+}^{\otimes{n}}\nonumber\\
    =&\bra{\bm 0}U^\dagger \prod\limits_{k=1}^p e^{-i \beta_k  H_x} [Ue^{-i \gamma_k (\ketbra{\bm 0}{\bm 0})}U^\dagger]\ket{+}^{\otimes{n}}\nonumber\\
    =&\bra{\bm 0} \prod\limits_{k=1}^p e^{-i \beta_k  H_x} e^{-i \gamma_k \ketbra{\bm 0}{\bm 0}}\ket{+}^{\otimes{n}}.
\end{align}
The first equality follows from $U\ket{\bm 0}=\ket{t}$ where $\ket{\bm 0} = \ket{0}^{\otimes n}$.
The second equality follows from the definition of the matrix exponential. The third equality follows as $U$ commutes with $H_x$ as does any analytic function of $H_x,$
and $U\ket{+}^{\otimes n}=\ket{+}^{\otimes n}.$ Thus, the overlap is seen to be target independent.

\end{proof}

\begin{remark}
Overlap invariance introduced in Proposition \ref{invar} shows that optimisation problems in Eqs.~\eqref{argmax} and~\eqref{equival} do not depend on the target. Therefore, optimal parameters are the same for any target state. 
Thus, with no loss of generality we limit our consideration to the target $\ket t = \ket {\bm 0}$.
\end{remark}


Preparation of state \eqref{ansatz} requires a strategy to assign $2p$ variational parameters by outerloop optimisation.

\begin{remark}[Global optimisation] A strategy when all $2p$ parameters are optimized simultaneously which might provide the best approximation to prepare $\ket{t}$. 
\end{remark}

\begin{remark}[Layerwise training]
Optimisation of parameters layer by layer. At each step after a layer is trained, all parameters are fixed. A new layer is added and only the parameters corresponding to the new layer are optimized. 
\end{remark}

Global optimisation is evidently challenging for high depth circuits. The optimisation can, in principle, be simplified by exploiting problem symmetries \cite{shaydulin2021exploiting} and leveraging parameter concentrations \cite{akshay2021parameter,streif2019comparison}. Layerwise training might avoid barren plateaus \cite{skolik2021layerwise} yet is known~\cite{campos2021training} to stagnate at some critical depth, past which additional layers (trained one at a time) do not improve overlap. Local coherent noise was found to re-establish the robustness of layerwise training \cite{campos2021training}.

\section{$p=1$ QAOA}
For a single layer, the global and layerwise strategies are equivalent. Such a circuit was considered to establish parameter concentrations \cite{akshay2021parameter} analytically. The overlap was shown to be:
\begin{align}
     \abs{g_1(\gamma,\beta)}^2 =\dfrac{1}{2^n} \big[1&+ 2 \cos^{n}{\beta}\left(\cos{(\gamma - n\beta)} - \cos{n\beta}\right) + 2\cos^{2n}{\beta}\left( 1 - \cos{\gamma} \right) \big]. 
    \label{algebraic_overlap}
 \end{align}
To find extreme points of \eqref{algebraic_overlap} the authors in \cite{akshay2021parameter} set the derivatives with respect to $\gamma$ and $\beta$ to zero. This approach leads to solutions which contain maxima but also minimum of the overlap \eqref{algebraic_overlap}. These must be carefully separated. Moreover, this approach ignores the operator structure of the overlap as presented here. For aesthetics, subscript {\it opt} in $\bm \gamma_{opt}$ and $\bm \beta_{opt}$ is further omitted. 

\begin{theorem}
\label{linear_relation}
Optimal $p=1$ QAOA parameters relate as $\gamma=\pi-2\beta$.
\end{theorem}
\begin{proof}
To maximize the absolute value of the overlap 
\begin{align}
    g \equiv g_1(\gamma,\beta)=  \bra{\bm 0}e^{-i \beta  H_x} e^{-i \gamma P}\ket{+}^{\otimes{n}},
    \label{p=1_overlap}
\end{align} 
with $P=\ketbra{\bm 0}{\bm 0}$ we use the standard conditions $\dfrac{\partial(gg^*)}{\partial\gamma}=\dfrac{\partial(gg^*)}{\partial\beta}=0$. Setting the first derivative to zero we arrive at 
\begin{align}
    \bra{\bm 0}e^{-i \beta H_x}e^{-i \gamma P} P\ket{+}^{\otimes{n}}g^*=\bra{+}^{\otimes{n}}Pe^{i \gamma P}e^{i \beta H_x}\ket{\bm 0}g.
    \label{d_gamma}
\end{align}
Using the explicit form of the projector and the fact that $\bra{\bm 0}e^{-i \beta H_x} \ket{\bm 0}=\cos^n\beta$ Eq.~\eqref{d_gamma} simplifies into 
\begin{align}
    g=g^*e^{-2i \gamma} \Leftrightarrow ge^{i\gamma} = g^*e^{-i\gamma},
    \label{g_gamma}
\end{align}
which is equivalent to 
\begin{align}
    \arg g=-\gamma.
    \label{arg_g}
\end{align}
Then the derivative  of expression \eqref{p=1_overlap} with respect to $\beta$  is set to zero and we arrive at
\begin{align}
    \bra{\bm 0}e^{-i \beta H_x}H_xe^{-i \gamma P} \ket{+}^{\otimes{n}}g^*=\bra{+}^{\otimes{n}}e^{i \gamma P}H_xe^{i \beta H_x}\ket{\bm 0}g.
\end{align}
Moving $H_x$ next to its eigenstate $\ket{+}^{\otimes{n}}$ is compensated as follows:
\begin{align}
    \bra{\bm 0}e^{-i \beta H_x}\{e^{-i \gamma P} H_x+(e^{-i \gamma}-1)[H_x,P]\} \ket{+}^{\otimes{n}}g^*\nonumber\\
    =\bra{+}^{\otimes{n}}\{H_xe^{i \gamma P}+[P,H_x](e^{i \gamma}-1)\}e^{i \beta H_x}\ket{\bm 0}g.
    \label{d_beta}
\end{align}
After simplification (see remark  \ref{trivial_solutions}) we arrive at
\begin{align}
    -g A =g^*A^*e^{-i\gamma},
    \label{A_gamma}
\end{align}
where $A=\bra{+}^{\otimes{n}}[P,H_x]e^{i \beta H_x}\ket{\bm 0}$.
Now $g^*$ is substituted from Eq.~\eqref{g_gamma} to establish
\begin{align}
    -e^{-i\gamma}A=A^*.
\end{align}
Thus, similar to Eq.~\eqref{arg_g} we arrive at
\begin{align}
    \arg A=\dfrac{\gamma+\pi}{2}.
    \label{arg_A}
\end{align}
$A$ is calculated as 
\begin{align}
      A\sqrt{2^n}=\bra{\bm 0}(H_x-n) e^{i\beta H_x}\ket{\bm 0} =
     -n\cos^{n-1}\beta e^{-i\beta},
\end{align}
which shows that $\arg A = \pi -\beta$.
Thus, from Eq.~\eqref{arg_A} we arrive at 
\begin{align}
     \pi-\beta= \dfrac{\gamma+\pi}{2},
\end{align}
which finally establishes $\gamma=\pi-2\beta$.
\end{proof}

\begin{remark}[Trivial solutions]
Eq.~\eqref{d_beta} has three pathological solutions which must be ruled out: (i) $\sin\dfrac{\gamma}{2}=0$ (which sets $e^{i\gamma}-1=0$), (ii) $\cos\beta=0$ (which sets $A=0$), (iii) $g(\gamma, \beta)=0$. 
All three cases imply $\abs{g(\gamma,\beta)}\le g(0,0)$.
\label{trivial_solutions}
\end{remark}

\begin{remark}
The zero derivative conditions result in \eqref{g_gamma} and \eqref{A_gamma} which have a similar form, viz.~$x=x^*e^{i\varphi}$. The first condition \eqref{g_gamma} can be obtained without differentiation  \cite{campos2021training} using the explicit form of the overlap Eq.~\eqref{p=1_overlap} 
\begin{align}
    g\sqrt{2^n}=e^{-i\gamma}\cos^n\beta + (e^{-i\beta n} -\cos^n\beta),
    \label{g_explicit}
\end{align}
and the fact that $\max_{\gamma} \abs{A e^{-i \gamma}+B} = \abs{A}+\abs{B}$ for any $A,B\in \mathbb{C}$. 
Although the derivative with respect to $\beta$ leads to the condition \eqref{A_gamma}, we find no way to recover this using elementary {\it alignment arguments}. 
\end{remark}

To find optimal parameters one needs to solve the zero derivative conditions and then take solutions that deliver a global maximum to the overlap. For convenience, we substitute $\gamma=\pi-2\beta$ to the overlap function \eqref{g_explicit}, square it and after simplification arrive at 
\begin{align}
    \abs{g}^2 2^n = 1+4\cos^{n+1}\beta(\cos^{n+1}\beta-\cos(n+1)\beta),
    \label{g2_beta}
\end{align}
which is used to prove the next theorem. 

\begin{theorem}
\label{opt_beta}
The optimal $p=1$ QAOA parameters converge as $\beta n\to\pi$ and $\gamma\to\pi$ when $n\to\infty$.
\end{theorem}

\begin{proof}
Using the explicit form of the overlap \eqref{g_explicit}, from Eq.~\eqref{g_gamma} one can establish
\begin{align}
    \operatorname{Im} [e^{i\gamma}(e^{-i\beta n} -\cos^n\beta)]=0.
\end{align}
Substituting $\gamma=\pi-2\beta$ one arrives at 
\begin{align}
    \operatorname{Im} [e^{-i(n+2)\beta} -e^{-2i\beta}\cos^n\beta)]=0,
\end{align}
which is equivalent to
\begin{align}
    \sin(n+2)\beta=\sin2\beta\cos^n\beta.
    \label{eq_beta}
\end{align}
We solve this equation in the limit $n\to\infty$. In this limit $\sin2\beta\cos^n\beta\to 0$ independent of the value of $\beta$. Thus, the left hand side of Eq.~\eqref{eq_beta} tends to zero. This implies that the leading order solution scales as 
\begin{align}
    \beta=\dfrac{k\pi}{n+2} + o(n^{-1})
    \label{beta_k}
\end{align}
where $k<n$ is a positive integer (in principle, $n$-dependent). 
To recover the optimal constant $k$ we substitute Eq. \eqref{beta_k} to Eq.~\eqref{g2_beta} to obtain
\begin{align}
    \abs{g}^2 2^n = 1+4\cos^{n+2}\dfrac{k\pi}{n+2}\Big(\cos^{n}\dfrac{k\pi}{n+2}-(-1)^k\Big)
    \label{g2_k}
\end{align}
up to $o(1)$ terms. Finally, as cosine is monotonously decreasing in the interval $[0,\pi)$ it is evident that the overlap maximizes for smallest odd constant $k=1$. 
Therefore, the optimal parameter $\beta$ is given by
\begin{align}
    \beta=\dfrac{\pi}{n+2}+o(n^{-1}) = \dfrac{\pi}{n}+o(n^{-1}),
\end{align}
which implies $n\beta\to\pi$ and thus $\gamma=\pi-2\beta\to\pi$ when $n\to\infty$.
\end{proof}
\
\begin{remark}
In theorem \ref{opt_beta} the leading order solutions were found for optimal parameters. Higher order corrections in n $n$ are found from Eq.~\eqref{eq_beta}. For example, it is straightforward to show that 
\begin{align}
\beta = \dfrac{\pi}{n}-\dfrac{4\pi}{n^2}+O(n^{-3}),
\label{final_beta}\\
\gamma = \pi-\dfrac{2\pi}{n}+\dfrac{8\pi}{n^2}+O(n^{-3}).
\label{final_gamma}
\end{align}
\end{remark}
\begin{remark}
Expressions \eqref{final_beta} and \eqref{final_gamma} were used to demonstrate parameter concentrations \cite{akshay2021parameter}, i.e. the effect when optimal parameters for $n$ and $n+1$ qubits are polynomially close.
\end{remark}

Theorems \ref{linear_relation} and \ref{opt_beta} presented above provide state of the art analytical results for state preparation with $p=1$ depth QAOA circuit. 
For deeper circuits and more general settings, analysis becomes complicated and known results are mostly numerical.
Therefore, below we provide a list of numerical effects for deeper circuits which lack analytical explanations.

\section{Empirical findings missing analytical theory}
\subsection{Parameter concentration in $p\ge2$ QAOA}
From expression \eqref{ansatz} overlaps for circuits of different depths are related recursively as
\begin{align}
g_{p+1}(\bm\gamma,\bm\beta,\gamma_{p+1},\beta_{p+1}) = g_p(\bm \gamma,\tilde{\bm\beta})+g_p(\bm\gamma,\bm\beta)\cos^n\beta_{p+1}(e^{-i \gamma_{p+1}}-1), 
\end{align}
where $\tilde {\bm \beta} = (\beta_1+\beta_{p+1},\dots,\beta_p+\beta_{p+1})$.
This recursion was used in \cite{akshay2021parameter} for $p=2$ where it was shown that in the thermodynamic limit $n\to \infty$ the zero derivative conditions let one obtain solutions for which $n\beta\to \pi$ and $\gamma\to\pi$.  This establishes parameter concentrations \cite{akshay2021parameter}.
The effect was further confirmed numerically up to $n=17$ qubit and $p=5$ layers \cite{akshay2021parameter}. For arbitrary depth, parameter concentrations are conjectured, yet analytical confirmation remains open.

\subsection{Last layer behaviour}
Theorem \ref{linear_relation} establishes the linear relation between optimal parameters independent of the number of qubits $n$. Using a global training strategy for the same problem with $p\ge 2$ depth circuits, it was numerically observed \cite{akshay2021parameter} that optimal parameters depend on the depth, yet usually can be approximately described by some linear relation. In the present work, we have observed that the last layer is distinctively characterized by the very same linear relation $\gamma_p+2\beta_p=\pi$ stated in Theorem \ref{linear_relation}.
We numerically confirmed this up to $p=5$ layers and $n=17$ qubits, as shown in figure~\ref{linearity}. The effect remains unexplained analytically and could be the manifestation of some hidden ansatz symmetry.  

\begin{figure}[!tbh]
   \centerline{\includegraphics[clip=true,width=0.8\linewidth]{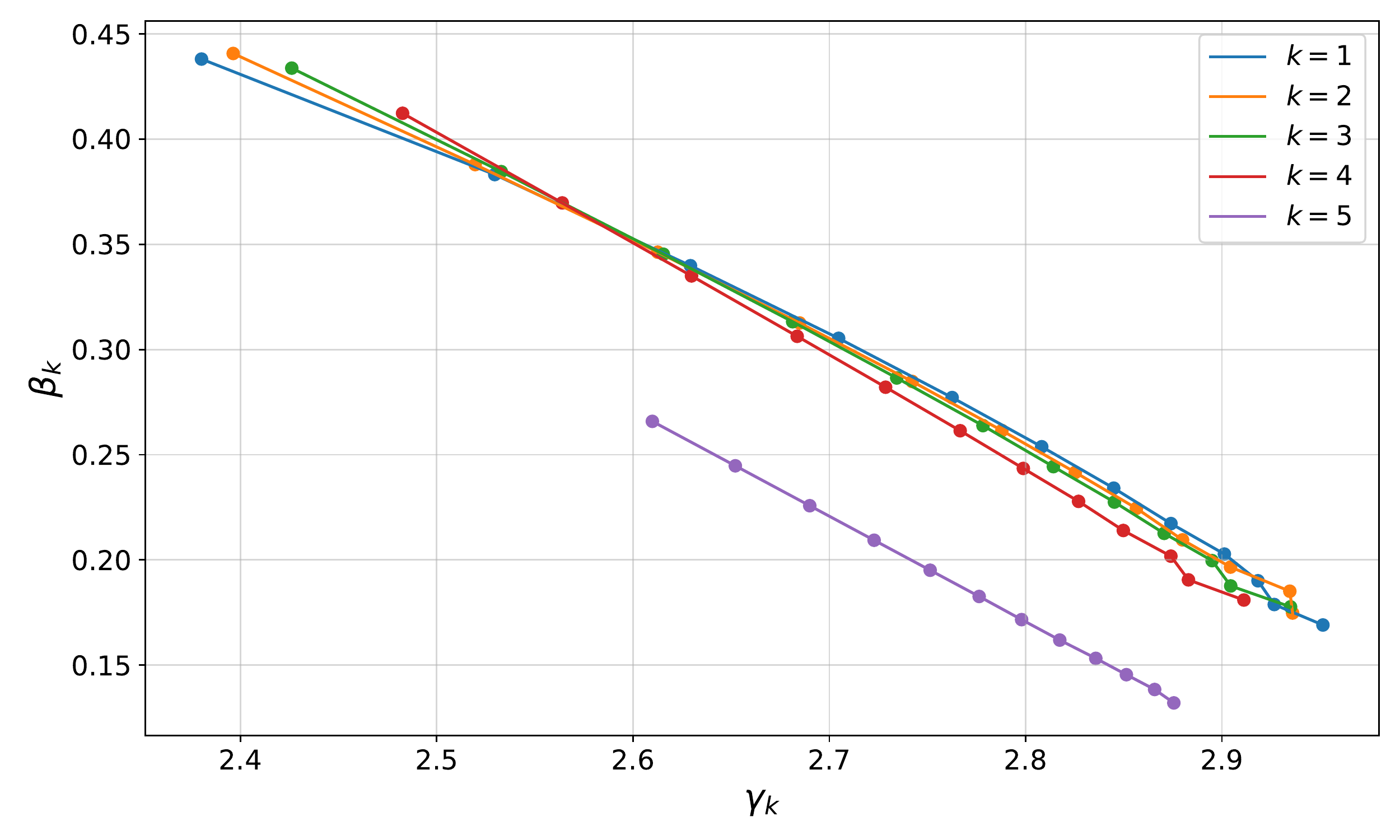}}
       \caption{Optimal angles of $p=5$ depth circuit for $n=6,\dots,17$. While the first layers can be approximately described by a linear relation, the last layer fits $\gamma_p+2\beta_p=\pi$. Moreover, the values of last layer's parameters are evidently distinct from the previous layers.}
    \label{linearity}
\end{figure}

\subsection{Saturation in layerwise training at $p=n$}
It was demonstrated \cite{campos2021training} that 
layerwise training {\it saturates}, meaning that past a critical depth $p^*$ overlap cannot be improved with further layer additions. 
Due to this effect, naive layerwise training performance falls below global training.
Training saturation in layerwise optimisation was reported in \cite{campos2021training} and confirmed up to $n=10$ qubits.
Most surprisingly, the saturation depth $p^*$ was observed to be equal to the number of qubits $n$. Two effects remain unexplained analytically.  Firstly does $p^* = n$.  Secondly, could one go beyond the necessary conditions in \cite{campos2021training} to explain saturations? 



\subsection{Removing saturation in layerwise training}
Any modification in the layerwise training process that violates the necessary saturation conditions can remove the system from its original saturation points. This idea was exploited in \cite{campos2021training}, where two types of variations were introduced for system sizes up to $n=7$: (i) undertraining the QAOA circuit at each iteration and (ii) training in the presence of random coherent phase noise. Whereas both modifications (i) and (ii) removed saturations at $p=n$ yet the reason remains unexplained.  


\section{Conclusion}
We have proven a relationship between optimal Quantum Approximate Optimisation Algorithm parameters for $p=1$ and in the thermodynamic limit, we recover optimal angles.  We demonstrated the effect of parameter concentrations for $p=1$ QAOA circuits using an operator formalism. Finally, we present a list of numerical effects, observed for particular system size and circuit depth, which are yet to be explained analytically. These unexplained effects include both limitations and advantages to QAOA.  While difficult, adding missing theory to these subtle effects would improve our understanding of variational algorithms.  

\section*{Acknowledgments}
The authors acknowledge support from the research project, {\it Leading Research Center on Quantum Computing} (agreement No.~014/20).

\section*{References}
\bibliography{refs.bib}

\bibliographystyle{unsrt}
\end{document}